\documentclass[preprint,12pt,dvipdfmx]{elsarticle}

\usepackage{amsmath,amssymb,amsthm}

\makeatletter
\def\ps@pprintTitle{%
 \let\@oddhead\@empty
 \let\@evenhead\@empty
 \def\@oddfoot{}%
 \let\@evenfoot\@oddfoot}
\makeatother

\usepackage{graphicx}          
\theoremstyle{plain}
\newtheorem{theorem}{Theorem}
\newtheorem{corollary}[theorem]{Corollary}

\newtheorem{lemma}[theorem]{Lemma}

\usepackage{bm}

\theoremstyle{definition}
\newtheorem{definition}[theorem]{Definition}

\theoremstyle{remark}

\biboptions{sort&compress}

\usepackage{subcaption}
\usepackage{siunitx}
\usepackage{setspace}
\usepackage{esvect}
\usepackage{url}

\DeclareSymbolFont{bbold}{U}{bbold}{m}{n}
\DeclareSymbolFontAlphabet{\mathbbold}{bbold}

\numberwithin{equation}{section}

\usepackage{accents}

\DeclareMathOperator*{\col}{col}

\usepackage{mathtools}

\usepackage{color}

\usepackage{ulem}
\definecolor{forestgreen}{rgb}{0.33,0.61,0.34}

\begin{document}

\begin{frontmatter}

\title{Analysis of the susceptible-infected-susceptible epidemic dynamics in networks via the non-backtracking matrix}

\author[bristol,ubmath,ubcdse]{Naoki Masuda\corref{cor}}
\ead{naokimas@buffalo.edu}
\fntext[bristol]{Department of Engineering Mathematics, University of Bristol, Woodland Road, Clifton, Bristol BS8 1UB, UK}
\fntext[ubmath]{Department of Mathematics, University at Buffalo, State University of New York, Buffalo, NY 14260, USA }
\fntext[ubcdse]{Computational and Data-Enabled Science and Engineering Program, University at Buffalo, State University of New York, Buffalo, NY 14260, USA}
\cortext[cor]{Corresponding author}

\author[penn,ApplMath_penn]{Victor M. Preciado}
\ead{preciado@seas.upenn.edu}
\fntext[penn]{Department of Electrical and Systems Engineering, University of Pennsylvania, Philadelphia, PA 19104, USA}
\fntext[ApplMath_penn]{Applied Mathematics \& Computational Science, University of Pennsylvania, Philadelphia, PA 19104, USA.}
\author[naist,osaka]{Masaki Ogura\corref{cor}}
\ead{m-ogura@ist.osaka-u.ac.jp}
\fntext[naist]{Division of Information Science, Nara Institute of Science and Technology, Ikoma, Nara 630-0192, Japan}
\fntext[osaka]{Graduate School of Information Science and Technology, Osaka University, Suita, Osaka 565-0871, Japan}

\begin{abstract}
%
%
We study the stochastic susceptible-infected-susceptible model of epidemic processes on finite directed and weighted networks with arbitrary structure. We present a new lower bound on the exponential rate at which the probabilities of nodes being infected decay over time. This bound is directly related to the leading eigenvalue of a matrix that depends on the non-backtracking and incidence matrices of the network. The dimension of this matrix is $N+M$, where $N$ and $M$ are the number of nodes and edges, respectively. We show that this new lower bound improves on an existing bound corresponding to the so-called quenched mean-field theory. Although the bound obtained from a recently developed second-order moment-closure technique requires the computation of the leading eigenvalue of an $N^2\times N^2$ matrix, we illustrate in our numerical simulations that the new bound is tighter, while being computationally less expensive for sparse networks. We also present the expression for the corresponding epidemic threshold in terms of the adjacency matrix of the line graph and the non-backtracking matrix of the given network.
\end{abstract}


\end{frontmatter}

\section{Introduction}

Epidemic processes are probably one of the most extensively studied dynamical processes in complex networks \cite{Keeling2005JRSocInterface,Barrat2008book,Masuda2013F1000,Pastorsatorras2015RevModPhys,KissMillerSimon2017book}. These processes can be used for modeling the spread of infectious diseases in contact networks, as well as news in (offline or online) social networks, or computer viruses in communication networks, to name a few applications. 
A fundamental question in the analysis of epidemic processes, in the case of both deterministic and stochastic models, is to quantify the total number of nodes being infected by the spread over time. In most epidemic models, we find two clearly differentiated dynamical phases: one phase in which an initial infection quickly dies out and another phase in which an infection may propagate to a large fraction of the network. The concept of epidemic threshold is used for characterizing the conditions separating these two dynamical phases.


Most of the existing stochastic epidemic models are Markov processes where the disease-free state is a unique absorbing state. This absorbing state is reached with probability one in finite time, regardless of the initial set of infected nodes or the values chosen for the parameters of the model. A critical distinction between the two phases described above is the expected time required to reach the disease-free state. In the first phase mentioned above, the epidemic dynamics converges exponentially fast towards the absorbing state. In contrast, in the second phase, this time can be exponentially long in terms of the number of nodes. It is also worth remarking that this observation is not applicable to stochastic epidemic processes taking place in infinite networks
\cite{Liggett1999book,Durrett2010PNAS} or deterministic models taking place in both finite and infinite networks
\cite{Pastorsatorras2001PhysRevLett,Barrat2008book,Pastorsatorras2015RevModPhys,KissMillerSimon2017book} because in both these cases it is possible for the disease to survive forever.
Therefore, for stochastic epidemic processes in finite networks,
the exponential rate at which the number of infected individuals decays toward zero, called the \textit{decay rate}, is a relevant characterization of the dynamics
\cite{Ganesh2005,Chakrabarti2008,VanMieghem2009a}. Intuitively, if the disease-free equilibrium takes a long time to be reached (in expectation), the decay rate would be close to zero. In contrast, if the infection dies out exponentially fast, the decay rate would be a positive value bounded away from zero. The decay rate can e.g. be used to measure the performance of control strategies aiming to eradicate an epidemic exponentially fast~\cite{Wan2008IET,Preciado2014,AbadTorres2016,Ogura2018c}. 

Generally speaking, finding the decay rate of stochastic epidemic processes in a large network is computationally hard. This is because
the number of possible states in the Markovian models typically used to model epidemics over networks grows exponentially in terms of the number of nodes in the network. Specifically, the decay rate is given by the leading eigenvalue of the transition-probability matrix of the Markovian model, whose dimension depends exponentially on the number of nodes. For example, in the case of the susceptible-infected-susceptible (SIS) model on $N$ nodes,
the exponential rate corresponds to the leading eigenvalue of a $2^N \times 2^N$ transition-rate  matrix~\cite{VanMieghem2009a}, which is computationally challenging to calculate for large networks. An alternative approach to the exact computation of the decay rate is to seek computationally feasible bounds. For example, for the SIS model, a lower bound can be obtained using a mean-field approximation. This approximation is based on a first-order moment-closure technique allowing us to compute a bound on the decay rate from the leading eigenvalue of an $N\times N$ matrix~\cite{Ganesh2005,Preciado2014}. However, this mean-field approximation can result in a loose bound for many networks~\cite{Ogura2018SystControlLett}. To increase the accuracy of the approximation, the authors proposed a tighter bound on the decay rate using a second-order moment-closure techniques~\cite{Ogura2018SystControlLett}. This tighter bound, however, requires the computation of the leading eigenvalue of an $N^2\times N^2$ matrix, which can be computationally prohibitive when analyzing epidemic processes in large networks.  

In the present work, we derive a new lower bound on the decay rate of the stochastic SIS model in an arbitrary finite network. This new bound depends on the leading eigenvalue of an $(N+M) \times (N+M)$ matrix, where $M$ is the number of directed edges; hence, for sparse networks --- such as networks with a bounded maximum degree --- the proposed lower bound is computationally more tractable than the bound derived in \cite{Ogura2018SystControlLett}. Our lower bound is based on an alternative second-order moment-closure technique aiming to overcome the computational challenges of existing second-order moment-closure techniques. The new bound depends on the non-backtracking matrix \cite{Hashimoto1989AdvStudPureMath,Alon2007CommContempMath}.
The non-backtracking matrix has recently gained popularity in the network science community because it is the basis of efficient and theoretically appealing techniques for community detection, network centralities, and others (see references in~\cite{Masuda2017PhysRep}).
We theoretically prove that the new lower bound 
is tighter than the first-order lower bound. We also show that our new lower bound is numerically more accurate than the bound obtained in~\cite{Ogura2018SystControlLett}.
We also present a new epidemic threshold, which corresponds to our lower bound on the decay rate. The new epidemic threshold is given in terms of the adjacency matrix of the line graph and the non-backtracking matrix of the given network.


\section{Problem statement}\label{sec:problemStatment}


We start with mathematical preliminaries. A directed graph is defined as the pair $\mathcal G = (\mathcal V, \mathcal E)$, where $\mathcal V =  \{v_1, \dotsc, v_N\}$ is a finite ordered set of nodes, $N$ is the number of nodes, and $\mathcal E \subset \mathcal V \times \mathcal V$ is a set of directed edges. By definition, $(v, v') \in \mathcal E$ indicates that there is an edge from $v$ to $v'$.
%
%
%
The adjacency matrix of~$\mathcal{G}$ is an $N\times N$ matrix of which
the $(i, j)$-th entry is equal to $1$ if $(v_i, v_j) \in \mathcal E$ and $0$ otherwise.
%
%
An in-neighbor of $v$ is a node $v^{\prime}$ such that $(v^{\prime}, v)\in \mathcal E$. 

We denote the identity and the zero matrices by $I$ and~$O$, respectively. 
A real matrix~$A$ (or a vector as its special case) is said to be nonnegative,
denoted by $A\geq 0$, if all the entries of~$A$ are nonnegative. If all the entries of~$A$ are positive, then $A$ is said to be positive. We say that $A\leq B$, where $A$ and $B$ are of the same dimension, whenever $B-A\geq 0$. A square matrix~$A$ is said to be Metzler if all its off-diagonal entries are nonnegative~\cite{Farina2000}. If $A$ is Metzler, it holds true that $e^{At} \geq
0$ for all $t\geq 0$ \cite{Farina2000}. For a Metzler matrix~$A$, the maximum real part of the eigenvalues of~$A$ is denoted
by~$\lambda_{\max}(A)$. For any matrix $A$, the spectral radius is the largest absolute value of its eigenvalues and denoted by $\rho(A)$.

We study the stochastic SIS model on networks, which is also known as the contact process in the probability theory literature \cite{Liggett1999book}. This model is defined as follows: let $\mathcal G = (\mathcal V, \mathcal E)$ be a directed graph. At any given continuous time~$t \geq 0$, each node is in one of the two possible states, namely, susceptible (i.e. healthy) or infected. An infected node $v_i$ stochastically transits to the susceptible state at a constant instantaneous rate of $\delta_i > 0$, which is called the \textit{recovery rate} of node~$v_i$. Whenever $v_i$ is susceptible
and its infected in-neighbor $v_j$ is infected, then $v_j$ stochastically and independently infects $v_i$ at a constant instantaneous rate of $\beta_{ji}$. We call $\beta_{ji} > 0$ the \textit{infection rate}. Note that the present SIS model effectively accommodates directed and weighted networks because the infection rate $\beta_{ij}$ is allowed to depend on $v_i$ and $v_j$.

The SIS model is a continuous-time Markov process with $2^N$ possible states
\cite{VanMieghem2009a,Pastorsatorras2015RevModPhys,KissMillerSimon2017book} and
has a unique absorbing state in which all the $N$ nodes are susceptible. Because this absorbing state is reachable from any other state, the dynamics of the SIS model reaches the disease-free absorbing equilibrium in finite time with probability one. The aim of the present paper is to study how fast this disease-free equilibrium is reached in expectation. This can be quantified via the following definition:
\begin{definition}
\label{defn:new}
Let $p_i(t)$ be the probability that the $i$th node is infected at time $t$. 
The \textit{decay rate} of the SIS model is defined by
\begin{equation}
\gamma = - \limsup_{t\to\infty} \frac{\log \sum_{i=1}^N p_i(t)}{t},
\end{equation}
where all nodes are assumed to be infected at $t=0$.
%
%
\end{definition}
Definition~\ref{defn:new} states that $\sum_{i=1}^N p_i(t)$, which is equal to the expected number of infected nodes at time $t$, roughly decays exponentially in time as $\propto e^{-\gamma t}$. Because the number of infected nodes always becomes zero in finite time, the SIS model always has a positive decay rate (potentially close to zero), even if the infection rate is large.


The decay rate has theoretically been studied
in continuous time \cite{VanMieghem2009a} and discrete time \cite{Chakrabarti2008} 
SIS models and is closely related to other quantities of interest, such as the epidemic threshold~\cite{VanMieghem2009a} and the
mean time to absorption~\cite{Ganesh2005}. However, exact computation of the decay rate is computationally demanding in
practice. Even in the homogeneous case, where all nodes share the same infection
and recovery rates, the decay rate equals the modulus of the largest real part
of the non-zero eigenvalues of a $2^N\times 2^N$ matrix representing the
infinitesimal generator of the Markov chain~\cite{VanMieghem2009a}. 

Due to the difficulty of its computation, several approaches have been proposed to bound the decay rate. A first-order lower bound, which corresponds to the so-called quenched mean-field approximation \cite{Pastorsatorras2015RevModPhys},
is derived as follows \cite{Ganesh2005,Preciado2014}:
Let $\bm p(t) = \left[p_1(t), \ldots, p_N(t)\right]^{\top}$, where $\top$ represents the matrix transposition. We define the $N\times N$ matrices $B$ and $D$ by
\begin{equation}
B_{ij} = 
\begin{cases}
\beta_{ij},& \mbox{if $(v_i, v_j)\in \mathcal E$,}
\\
0,&\mbox{otherwise,}
\end{cases}
\end{equation}
and
\begin{equation}
D = \text{diag}(\delta_1, \ldots, \delta_N),
\end{equation}
where $\text{diag}(\alpha_1, \ldots, \alpha_N)$ is the $N\times N$ diagonal matrix whose diagonal elements are equal to $\alpha_1$, $\ldots$, $\alpha_N$. Note that matrix $B$ fully contains the information about the adjacency matrix of $\mathcal G$. Then, one can show that $\bm p(t)
\leq e^{(B^\top-D)t} \bm p(0)$, which implies that
\begin{equation}
\label{eq:rho1}
\gamma \geq \gamma_1 \equiv -\lambda_{\max}(B^\top - D),
\end{equation}
where we will call $\gamma_1$ the first-order lower bound.
Although this lower bound is computationally efficient to find, 
there can be a large discrepancy between $\gamma_1$ and the true decay rate $\gamma$~\cite{Ogura2018SystControlLett}.

A second lower bound on the decay rate was proposed in a recent study \cite{Ogura2018SystControlLett} and summarized in
%
%
\ref{sec:Ogura2018 bound}.
This second bound depends on the leading eigenvalue of an $N^2 \times N^2$ matrix, which is computationally demanding when $N$ is relatively large. In this paper, we propose an alternative lower bound on the decay rate that is computationally efficient, provably more accurate than the first-order bound, and numerically tighter than the second bound described in \ref{sec:Ogura2018 bound}.

\section{Main results} \label{sec:main}

\subsection{A lower bound on the decay rate}

To state our mathematical results, we label the directed edges of a given network $\mathcal{G}$ as $\{e_1, \ldots, e_M\}$, where the $\ell$th edge ($1\le \ell\le M$) is represented by $e_\ell = (i_\ell, j_\ell)$, i.e., the edge is directed from node $v_{i_\ell}$ to node $v_{j_\ell}$. Although we use notations $i$ and $j$ to represent general nodes $v_i$ and $v_j$ in the following text, $i$ and $j$ with a subscript will exclusively represent the starting and terminating nodes of an edge, thus avoiding confusions. Define the \textit{incidence} matrix $C\in \mathbb{R}^{N\times M}$ of the network $\mathcal{G}$ by \cite{Wilson1972-2010book,Newman2010book}
\begin{equation}
C_{i\ell} = \begin{cases}
1, & \mbox{if $j_\ell = i$}, 
\\
-1, & \mbox{if $i_\ell = i$}, 
\\
0, & \mbox{otherwise}.
\end{cases}
\end{equation}
Also, define the \textit{non-backtracking} matrix $H\in \mathbb{R}^{M\times M}$ of $\mathcal{G}$ by \cite{Hashimoto1989AdvStudPureMath,Alon2007CommContempMath}
\begin{equation}
H_{\ell m} = \begin{cases}
1, & \mbox{if $j_\ell = i_m$ and $j_m \neq i_\ell$}, 
\\
0, & \mbox{otherwise.}
\end{cases}
\end{equation}

The main result of the present paper is stated as follows:
\begin{theorem}
\label{thm:main}
Define the $(N+M)\times (N+M)$ Metzler matrix 
\begin{equation}\label{eq:defMathcalA}
\mathcal A = \begin{bmatrix}
-D & C_+ B'
\\
D_2' C_-^\top & H^\top B' - B' - D_1' - D_2'
\end{bmatrix}, 
\end{equation}
where
\begin{align}\label{eq:defB'D'}
B' =& \text{diag}(\beta_{i_1 j_1}, \ldots, \beta_{i_M j_M}),\\
D_1' =& \text{diag}(\delta_{i_1}, \ldots, \delta_{i_M}),\\ 
D_2' =& \text{diag}(\delta_{j_1}, \ldots, \delta_{j_M}),
\end{align}
$C_+ = \max(C, 0)$, and $C_- = \max(-C, 0)$; $C_+$ and $C_-$ denote the positive and negative parts of the incidence matrix~$C$, respectively.
Then, we obtain the following lower bound on the decay rate:
\begin{equation}
\label{eq:bound}
\gamma \geq \gamma_2 \equiv -\lambda_{\max}(\mathcal A). 
\end{equation}
\end{theorem}

\begin{proof}
Define the binary variable $x_i(t)$ such that $x_i(t)=0$ or $x_i(t)=1$ if node~$v_i$ is susceptible or infected at time~$t$, respectively. The variables $x_1(t)$, $\ldots$, $x_N(t)$ obey a system of stochastic differential equations with Poisson jumps, and
the expectation $p_i(t) = E[x_i(t)]$ obeys
\begin{align}
\frac{dp_i}{dt} 
&= \left( \sum_{j=1}^N E[(1-x_i)x_j]\beta_{ji} \right) - \delta_i E[x_i]\notag\\
&
= \left( \sum_{j=1}^N \beta_{ji} q_{ji} \right) - \delta_i p_i,
\label{eq:dot p_i}
\end{align}
where 
\begin{equation}\label{eq:def:qij}
q_{ji}(t) = E[x_j(t)(1-x_i(t))]
\end{equation}
is equal to the joint probability that node $v_j$ is infected and node $v_i$ is susceptible at time~$t$.

Using the identities
\begin{equation}
\sum_{j=1}^N \beta_{ji} q_{ji} 
= 
\sum_{\ell=1; j_\ell = i}^{M} \beta_{i_\ell j_\ell} q_{i_\ell j_\ell} = \sum_{\ell=1}^{M}  [C_+]_{i\ell} 
[B']_{\ell\ell} q_{i_\ell j_\ell},
\end{equation}
one obtains 
\begin{equation}
\frac{dp_i}{dt} 
= \left(\sum_{\ell=1}^{M}  [C_+ B']_{i\ell} q_{i_\ell j_\ell} \right) - \delta_i p_i \quad (i \in \{1, \ldots, N\}).
\label{eq:dp_i/dt}
\end{equation}
Equation~\eqref{eq:dp_i/dt} is equivalent to
\begin{equation}
\label{eq:dpdt}
\frac{d\bm p}{dt} = C_+ B' \bm q - D\bm p,
\end{equation}
where we remind that $\bm p(t) = \left[p_1(t), \ldots, p_N(t)\right]^{\top}$ and define
\begin{equation}
\bm q(t) \equiv \left[ q_{i_1 j_1}(t), \ldots,
q_{i_M j_M}(t) \right]^{\top}.
\end{equation}
Using the notation~$p_{ij}(t) \equiv E[x_i(t) x_j(t)]$, one obtains
\begin{align}
\frac{dq_{ij}}{dt} 
&= 
- \left(\sum_{k=1}^N E[x_i (1-x_j) x_k] \beta_{kj}\right) + \delta_j E[x_ix_j]\notag\\ 
&\quad + \left(\sum_{k=1}^N E[(1-x_i) (1-x_j)x_k] \beta_{ki}\right) - \delta_i E[x_i (1-x_j)] \notag\\
&\leq 
- \beta_{ij} q_{ij} + \delta_j p_{ij}
+ \left(\sum_{k=1; k\neq j}^N \beta_{ki} q_{ki}\right) - \delta_i q_{ij}. 
\label{eq:20}
\end{align}
The first term on the right-hand side of the first line in Eq.~\eqref{eq:20} represents the rate at which node $v_j$ is infected when node $v_i$ is infected and node $v_j$ is susceptible; the second term represents the rate at which $v_j$ recovers when both $v_i$ and $v_j$ are infected; the third term represents the rate at which $v_i$ is infected when both $v_i$ and $v_j$ are susceptible; the fourth term represents the rate at which $i$ recovers when $v_i$ is infected and $v_j$ is susceptible.
To derive the last inequality in Eq.~\eqref{eq:20}, for the first term on the right-hand side, we ignored
all the $k$ values but $k=i$ in the summation and used $x_i^2 = x_i$. For the third term on the right-hand side, we used
$E[(1-x_i) (1-x_j)x_k] \le E[(1-x_i) x_k]$.

By combining Eq.~\eqref{eq:20} and
$p_{ij} = E[x_ix_j] = E[x_i] - E[x_i(1-x_j)] = p_i  - q_{ij}$, one obtains
\begin{equation}
\frac{dq_{i_\ell j_\ell}}{dt} 
\leq 
- (\beta_{i_\ell j_\ell} + \delta_{i_\ell} +
\delta_{j_\ell}) q_{i_\ell j_\ell} + \delta_{j_\ell} p_{i_\ell} + \sum_{k=1; k\neq j_\ell}^N \beta_{ki_\ell} q_{k i_\ell}.
\label{eq:dot q_ij}
\end{equation}
By combining Eq.~\eqref{eq:dot q_ij} with
\begin{equation}
\label{eq:Cpell}
[C_-^\top \bm p]_\ell 
%
%
= \sum_{i=1}^N [C_-]_{i \ell } p_i
= p_{i_\ell}
\end{equation}
and
\begin{equation}
\label{eq:Hqell}
\begin{multlined}
[H^\top B' \bm q]_\ell
= \sum_{m=1}^{M} H_{m\ell} B'_{mm} q_m
= \sum^M_{\substack{m=1; j_m = i_\ell,\\ j_\ell \neq i_m}} \beta_{i_m j_m} q_{i_m j_m}\\
= \sum_{k=1; k\neq j_\ell}^N \beta_{k i_\ell} q_{k i_\ell},
\end{multlined}
\end{equation}
one obtains
\begin{align}
\frac{dq_{i_\ell j_\ell}}{dt} 
\leq& 
-([B']_{\ell\ell} + [D'_1]_{\ell\ell} + [D'_2]_{\ell\ell}) q_{i_\ell j_\ell} + [D'_2]_{\ell\ell}
[C_-^\top \bm p]_\ell  + [H^\top B' \bm q]_\ell \notag\\
=&
-[B' \bm q]_{\ell} - [D'_1 \bm q]_{\ell} - [D'_2 \bm q]_{\ell}
+ [D'_2 C_-^\top \bm p]_\ell  + [H^\top B' \bm q]_\ell. 
\end{align}
By stacking this inequality with respect to $\ell$, one observes that there exists an $\mathbb{R}^{M}_+$-valued function $\bm \epsilon (t)$ defined for $t\in [0, \infty)$ such that 
\begin{equation}\label{eq:dq/dt}
\frac{d \bm q}{dt} = D_2' C_-^\top \bm p + (H^\top B' - B' - D_1' - D'_2) \bm q - \bm \epsilon. 
\end{equation}

Equations~\eqref{eq:dpdt} and~\eqref{eq:dq/dt} imply
\begin{equation}\label{eq:pSIS}
\frac{d}{dt}\begin{bmatrix}
\bm p\\ \bm q
\end{bmatrix}
=
\mathcal A \begin{bmatrix}
\bm p\\ \bm q
\end{bmatrix} - \begin{bmatrix}
\bm 0 \\ \bm \epsilon
\end{bmatrix}. 
\end{equation}
Because $\mathcal A$ is Metzler and $\bm \epsilon(t)$ is entry-wise nonnegative for every
$t\geq 0$, one obtains
\begin{align}
\begin{bmatrix}
\bm p(t) \\ \bm q(t)
\end{bmatrix}
&= 
e^{\mathcal A t} 
\begin{bmatrix}
\bm p(0) \\ \bm q(0)
\end{bmatrix}
-
\int_0^t
e^{\mathcal A(t-\tau)}
\begin{bmatrix}
\bm 0 \\ \bm \epsilon(\tau)
\end{bmatrix}\,{\rm d}\tau \notag\\
&
\leq e^{\mathcal
A t}\begin{bmatrix}
\bm p(0) \\ \bm q(0)
\end{bmatrix}, 
\end{align}
which proves Eq.~\eqref{eq:bound}.
\end{proof}

Next, to prove that the new lower bound is tighter than the first-order lower bound,
we start by stating (and proving) a convenient adaptation of the classical Perron--Frobenius theorem \cite{Hogben2006} for nonnegative matrices to the case of Metzler matrices.
\begin{lemma}\label{lem:aPF}
Let $M$ be an irreducible Metzler matrix. 
\begin{enumerate}
\item \label{item:aPF:eq} There exists a positive vector~$\bm v$ such that $M \bm v = \lambda_{\max}(M)\bm  v$.
\item \label{item:aPF:ineq}Assume that there exist a real number $\mu$ and a nonzero vector $\bm u\geq 0$ such that $M\bm u \leq \mu \bm u$ and $M\bm u\neq \mu \bm u$. Then, $\lambda_{\max}(M)< \mu$. 
\end{enumerate}
\end{lemma}
\begin{proof}
Let $\nu$ be a real number such that matrix $M' = M + \nu I$ is nonnegative. Note that the spectral radius of $M^{\prime}$ satisfies $\rho(M') = \lambda_{\max}(M)+\nu$. To prove the first statement, we use the Perron--Frobenius theorem (see Fact~5.b in \cite[Chapter~9.2]{Hogben2006}), which guarantees that $M'$ has a positive eigenvector~$\bm v$ corresponding to the eigenvalue $\rho(M')$. The vector $\bm v$ satisfies $M\bm v = M'\bm v - \nu \bm v = \left[\rho(M')-\nu\right] \bm v = \lambda_{\max}(M) \bm v$.

To prove the second statement, assume that a nonzero vector $\bm u\geq 0$ satisfies $M\bm u \leq \mu \bm u$ and $M\bm u\neq \mu \bm u$. Then, the nonnegative and irreducible matrix $M'$ satisfies $M' \bm u \leq (\mu + \nu) \bm u$ and $M' \bm u \neq (\mu+\nu) \bm u$. Therefore, the Perron--Frobenius theorem (see Fact~7.b in \cite[Chapter~9.2]{Hogben2006}) guarantees that $\rho(M')< \mu+\nu$, which yields $\lambda_{\max}(M) = \rho(M')-\nu < \mu$.
\end{proof}

The following theorem proves that the bound proposed in Eq.~\eqref{eq:bound} improves the first-order bound given by Eq.~\eqref{eq:rho1}.

\begin{theorem}\label{propORtheorem:strictInequality}
If the network is strongly connected, then $\gamma_2 > \gamma_1$.
\end{theorem}

\begin{proof}
Lemma~\ref{lem:aPF}.\ref{item:aPF:eq} implies that the irreducible Metzler matrix $B^\top - D$ has a positive eigenvector $\bm v$
corresponding to the eigenvalue $-\gamma_1$, i.e.,
\begin{equation}\label{eq:WeNeedTranspose}
(B^\top - D) \bm v = - \gamma_1 \bm v. 
\end{equation}
Define the positive $(N+M)$-dimensional vector $\bm \xi$ as
\begin{equation}
\bm \xi = \begin{bmatrix}
\bm v\\ \bm w
\end{bmatrix},\ \bm w = C_-^\top \bm v=\left[ v_{i_1}, \ldots, v_{i_M}\right]^{\top}. 
\end{equation} 
Let us define $\bm \zeta \equiv \mathcal A \bm \xi$ and decompose $\bm \zeta$ as 
\begin{equation}
\bm \zeta = \begin{bmatrix}
\bm \zeta_1 \\ \bm \zeta_2
\end{bmatrix}, 
\end{equation}
where $\bm \zeta_1$ and $\bm \zeta_2$ are $N$- and $M$-dimensional vectors, respectively. Simple algebraic manipulations yield
\begin{equation}\label{eq:Btop}
B^\top = C_+ B' C_-^\top.
\end{equation}
Therefore, one obtains 
$C_+ B' \bm w = C_+ B' C_-^\top \bm v = B^\top\bm v$. This implies that
\begin{align}
\bm \zeta_1 &= -D \bm v + C_+ B' \bm w \notag\\
&=
-D \bm v + B^\top \bm v \notag\\
&=
- \gamma_1 \bm v.
\label{eq:zeta_1}
\end{align} 
One also obtains
\begin{align}
\bm \zeta_2 &= D_2'C_-^\top \bm v + (H^\top B' - B' - D_1' - D_2') \bm w
\notag\\
&= (H^\top B' - B' - D_1') \bm w
\notag\\
&\leq (A_{L(\mathcal G)}^\top B'-B'-D_1')\bm w,
\label{eq:eta2}
\end{align}
where $A_{L(\mathcal G)}$ denotes the adjacency matrix of the line graph $L(\mathcal G)$ defined by 
\begin{equation}
\label{eq:def:ALG}
[A_{L(\mathcal G)}]_{\ell m} = \begin{cases}
1,& \mbox{if {$j_\ell = i_m$}} , 
\\
0,& \mbox{otherwise}.
\end{cases}
\end{equation}
Matrix $A_{L(\mathcal G)}$ satisfies 
\begin{equation}
\label{eq:C-Tc+}
A_{L(\mathcal G)}^{{\top}} = C_-^\top C_+
\end{equation}
because 
\begin{align}
[C_-^\top C_+]_{\ell m} 
&= \sum_{i=1}^N [C_-]_{i\ell} [C_+]_{im}
\notag\\
&=
\begin{cases}
1,&\mbox{if $i_\ell = j_m$},\\
0,&\mbox{otherwise.}
\end{cases}
\end{align}
Because simple algebraic manipulations yield $C_-^\top D = D_1' C_-^\top$, using Eqs.~\eqref{eq:WeNeedTranspose}, \eqref{eq:Btop}, and \eqref{eq:C-Tc+}, one obtains
\begin{align}
(A_{L(\mathcal G)}^\top B' - D_1') \bm w
&= 
C_-^\top C_+ B' C_-^\top \bm v -  D_1' C_-^\top \bm v 
\notag\\
&=
C_-^\top (B^\top -  D) \bm v
\notag\\
&= -\gamma_1 C_-^\top  \bm v
\notag\\
&=
-\gamma_1  \bm w.
\label{eq:ww}
\end{align}
Using Eqs.~\eqref{eq:eta2} and \eqref{eq:ww}, one obtains
\begin{equation}
\label{eq:for eta2 new}
\bm \zeta_2 \leq -\gamma_1 \bm w - B' \bm w.
\end{equation}
Because any entry of $B' \bm w$ is positive, Eqs.~\eqref{eq:zeta_1} and
%
%
\eqref{eq:for eta2 new} guarantee that positive vector~$\bm \xi$ satisfies $\mathcal A \bm \xi \leq - \gamma_1 \bm \xi$ and $\mathcal A \bm \xi \neq - \gamma_1 \bm \xi$. Because $\mathcal A$ is irreducible, as will be shown later, Lemma~\ref{lem:aPF} guarantees that $\lambda_{\max}(\mathcal A) < - \gamma_1$, which implies that $\gamma_2 > \gamma_1$.

Finally, let us show the irreducibility of matrix~$\mathcal A$ or, equivalently, the irreducibility of $\mathcal A^\top$. We regard the matrix $\mathcal A^\top$ as the adjacency matrix of a directed graph on $N+M$ nodes denoted by $\mathcal G'$. We label the nodes of $\mathcal G'$ as $p_1$, $\ldots$, $p_N$, $q_{i_1 j_1}$, $\ldots$, $q_{i_M j_M}$. The first term on the right-hand side of Eq.~\eqref{eq:dot p_i} implies that $\mathcal G'$ has an edge $(q_{i_\ell j_\ell}, p_{j_\ell})$ for all $\ell$, which corresponds to
$C_+ B'$ in Eq.~\eqref{eq:defMathcalA}.
The second term on the right-hand side of Eq.~\eqref{eq:dot q_ij} implies that $\mathcal G'$ has an edge $(p_{i_\ell}, q_{i_\ell j_\ell})$ for all $\ell$, which corresponds to $D_2^{\prime} C_-^{\top}$ in Eq.~\eqref{eq:defMathcalA}.

To show that $\mathcal G'$ is strongly connected, we first consider an arbitrary ordered pair of nodes $p_i$ and $p_j$ in $\mathcal G'$. Let us take a path $v_i = v_{\iota(0)}$, $v_{\iota(1)}$, \dots, $v_{\iota(s)}=v_j$ in the original graph $\mathcal G$. Then, from the above observation, we see that the graph $\mathcal G'$ contains the path $p_{i} = p_{\iota(0)}$, $q_{\iota(0)\iota(1)}$, $q_{\iota(1)}$, $q_{\iota(1)\iota(2)}$, \dots, $q_{\iota(s-1)\iota(s)}$, $p_{\iota(s)}=p_j$. 
Likewise, for an arbitrary ordered pair of nodes $p_i$ and $q_{i_\ell j_\ell}$ in $\mathcal G^{\prime}$, there is a path in $\mathcal G^{\prime}$ from $p_i$ to $p_{i_\ell}$. By appending edge $(p_{i_\ell}, q_{i_\ell j_\ell})$ to the end of this path, one obtains a path from $p_i$ to $q_{i_\ell j_\ell}$. A path from arbitrary $q_{i_\ell j_\ell}$ to $p_j$ and
one from $q_{i_\ell j_\ell}$ to  $q_{i_{\ell^{\prime}} j_{\ell^{\prime}} }$ can be similarly constructed. Therefore, a path exists between any pair of nodes in~$\mathcal G'$.
\end{proof}

\subsection{Epidemic threshold}

In this section, we assume that $\beta_{ij} = \beta$ and $\delta_i = \delta$, where $i, j \in \{1, \ldots, N\}$ and $\beta, \delta > 0$, and derive conditions under which the expected number of infected individuals decays exponentially fast. It holds true that having $\gamma_1 < 0$ in Eq.~\eqref{eq:rho1} is equivalent to the well-known epidemic threshold
$\beta/\delta > 1/  \lambda_{\max}(A)$ \cite{Wang2003SRDS,Chakrabarti2008,Pastorsatorras2015RevModPhys,KissMillerSimon2017book}.

Likewise, Theorem~\ref{thm:main} provides a tighter epidemic threshold given by
%
\begin{equation}\label{eq:extCondHetero:updated}
\left(\beta/\delta\right)_{\rm c} = \max\{\beta/\delta \mid \gamma_2 \ge 0\},
\end{equation}
where $\gamma_2$ is defined in Eq.~\eqref{eq:bound}.
Our following corollary provides an explicit expression of the epidemic threshold in terms of the adjacency matrix of the line graph and the non-backtracking matrix $H$.
\begin{corollary}
\label{cor:new}
Define the matrix $A_{L(\mathcal G)}$ by Eq.~\eqref{eq:def:ALG}. Then, 
\begin{equation}
\label{eq:epithre:updated}
\left(\frac{\beta}{\delta}\right)_{\rm c} = \frac{2}{\rho(A_{L(\mathcal G)}+H)-1}. 
\end{equation} 
\end{corollary}

\begin{proof}
We decompose $\mathcal A$ such that
\begin{equation}
\mathcal A = R + P, 
\end{equation}
where
\begin{equation}
R = \begin{bmatrix}
-\delta I & O
\\
\delta C_-^\top & -\beta I - 2\delta I
\end{bmatrix}
\end{equation}
and
\begin{equation}
P  = \begin{bmatrix}
O & \beta C_+
\\
O & \beta H^\top
\end{bmatrix}.
\end{equation}
Matrix~$R$ is a Metzler matrix, all the eigenvalues of $R$ have negative real parts, and matrix~$P$ is nonnegative. Therefore, Theorem~2.11 in Ref.~\cite{Damm2003LinAlgItsAppl} implies that $\lambda_{\max}(\mathcal A) < 0$ if and only if $\rho(R^{-1} P ) < 1$. Because
\begin{align}
R^{-1}P = & 
\begin{bmatrix}
-\dfrac{1}{\delta}I & O\\
-\dfrac{1}{\beta+2\delta }C_-^{\top} & - \dfrac{1}{\beta + 2\delta}I
\end{bmatrix}
P\notag\\
= &
\begin{bmatrix}
O & -\dfrac \beta \delta C+ \\ O & -\dfrac \beta{\beta + 2\delta}(C_-^\top C_+ + H^\top)
\end{bmatrix},
\end{align}
one obtains
\begin{align}\label{eq:rhoR-1P}
\rho(R^{-1} P) &= \frac{\beta}{\beta+2\delta}\rho(C_-^\top C_+ + H^\top) \notag\\
&= \frac{\beta}{\beta+2\delta }\rho(A_{L(\mathcal G)}^\top + H^\top)
\notag\\
&=
\frac{\beta}{\beta+2\delta }\rho(A_{L(\mathcal G)} + H),
\end{align}
where we used Eq.~\eqref{eq:C-Tc+}.
Therefore, $\rho(R^{-1} P) < 1$ if and only if
\begin{equation}
\frac{\beta}{\delta } < \frac{2}{\rho(A_{L(\mathcal G)} + H) - 1},  
\end{equation}
which is equivalent to Eq.~\eqref{eq:epithre:updated}.
\end{proof}

Remark: Corollary~\ref{cor:new} does not require strong connectedness (i.e., irreducibility of the adjacency matrix) of the network.

\section{Numerical results}

In this section, we carry out numerical simulations of the stochastic SIS dynamics for several networks to assess the tightness of the different lower bounds on the decay rate. In the following numerical simulations, we set $\beta_{ij} = \beta$ and $\delta_i = \delta$, where $i, j \in \{1, \ldots, N\}$, for simplicity.
We further assume that $\delta = 1$ without loss of generality (because changing $\beta$ and $\delta$ simultaneously by the same factor is equivalent to rescaling the time variable without changing $\beta$ or $\delta$).

We ran the stochastic SIS dynamics $10^4$ times starting from the initial condition in which all nodes are infected. For each run of the simulations, we measured the number of infected individuals at every integer time (including time 0) until the infection dies out or the maximum time, which is set to $5\times 10^4$, is reached. Then, at each integer time, we summed the number of infected nodes over all the runs and divided it by $N$ and by the number of runs ($= 10^4$), thus obtaining the average fraction of infected nodes, i.e., $\rho(t)\equiv \sum_{i=1}^N p_i(t)/N$, where $t=0, 1, \ldots$.

We calculated the decay rate from the observed $\{ \rho(t) : t=0, 1, \ldots\}$ as follows: because the fluctuations in $\rho(t)$ are expected to be large when $\rho(t)$ is small, we identified the smallest integer time at which $\rho(t)$ is less than $10^{-4}$ for the first time, and discarded $\rho(t)$ at this and all larger $t$ values. Then, because $\rho(t)$ is expected to decay exponentially in $t$, we calculated a linear regression between $\log \rho(t)$ and $t$ at the remaining integer values of $t$. The sign-flipped slope of this regression provides a numerical estimate of the decay rate. We confirmed that the Pearson correlation coefficient in the linear regression was at least $0.958$ for all networks and all $\beta$ values. The Pearson correlation was typically larger than $0.99$.

We used eight undirected and unweighted networks to compare the numerically obtained decay rate and the rigorous lower bounds. The lower bounds to be compared are $\gamma_1$, $\gamma_2$, and the one obtained in our previous study \cite{Ogura2018SystControlLett}, which is denoted by $\gamma_2^{\prime}$ (see \ref{sec:Ogura2018 bound} for a summary). 

Four of the eight networks used were created by generative models with $N=100$ nodes.
First, we generated a regular random graph in which all nodes had degree six, resulting in $300$ undirected edges (therefore, $M=600$ directed edges). Second, we used the Barab\'{a}si-Albert (BA) model to generate a power-law degree distribution with an exponent of 3 when $N$ is large \cite{Barabasi1999Science}.  We set the parameters $m_0=3$ and $m=3$, where $m_0$ is the initial number of nodes forming a clique in the process of growing a network, and $m$ is the number of edges that each new node initially brings into the network. With these parameter values,  the mean degree is approximately equal to $2m = 6$. The generated network had $294$ undirected edges. Third, we used a cycle graph, where each node had degree two (by definition), and there were $100$ undirected edges.
These three models lack community structure that many empirical contact networks have.
Therefore, as a fourth network, we used the Lancichinetti--Fortunato--Radicchi (LFR) model that can generate networks with community structure \cite{Lancichinetti2009PhysRevE-benchmark}.
The LFR model creates networks having a heterogeneous degree distribution and a heterogeneous distribution of community size. 
A small value of parameter $\mu$ corresponds to a strong community structure. We set $\mu=0.1$. We set the mean degree to six, the largest degree to $N/4 = 25$, the power-law exponent for the degree distribution to two and the power-law exponent for the distribution of community size to one. The network had $319$ undirected edges.

We also used four real-world networks, for which we ignored the direction and weight of the edge.
First, we used the dolphin social network, which has
$N=62$ nodes and $159$ undirected edges \cite{Lusseau2003BehavEcolSociobiol}.
A node represents a bottleneck dolphin individual. An edge indicates frequent association between two dolphins. This network is a connected network. 
Second, we used the largest connected component (LCC) of a coauthorship network of researchers in network science, which has
$N=379$ nodes and $914$ undirected edges \cite{Newman2006PhysRevE-collabo}.
A node represents a researcher publishing in fields related to network science. An edge indicates that two researchers have coauthored a paper at least once. Third, we used the LCC of an email network, which has $N=1,133$ nodes and $5,451$ undirected edges \cite{GuimeraDanon2003PhysRevE}. A node represents a member of the University Rovira i Virgili, Tarragona, Spain. An edge is an email exchange relationship between a pair of members.
Fourth, we used the LCC of the hamsterter network, which has
$N=1,788$ nodes and $12,476$ undirected edges
\cite{Kunegis2013%
%
%
}.
A node represents a user of the website hamsterster.com. An edge is a friendship relationship between two users.

For a range of values of $\beta$, we compare decay rates obtained numerically with the three lower bounds described in this paper for the eight networks mentioned above. The results are shown in Fig.~\ref{fig:decay rate}.
It should be noted that the decay rate and its bounds are equal to one for $\beta=0$ because we set $\delta=1$ for normalization. The bound $\gamma_2$ proposed in the present study is considerably tighter than the first-order bound, $\gamma_1$, for some networks, in particular, the cycle
(Fig.~\ref{fig:decay rate}(c)). The improvement tends to be more manifested for smaller networks. We also find that $\gamma_2$ is tighter than $\gamma_2^{\prime}$ for all the networks and infection rates, despite that $\gamma_2$ is easier to calculate than $\gamma_2^{\prime}$. For example, for the regular random graph
(Fig.~\ref{fig:decay rate}(a)) and the cycle (Fig.~\ref{fig:decay rate}(c)), $\gamma_2$ is close to the numerically estimated decay rate for small to moderate values of $\beta$, which is not the case for $\gamma_2^{\prime}$ as well as for $\gamma_1$.

\section{Conclusions}

We have introduced a lower bound on the decay rate of
the SIS model on arbitrary directed and weighted networks. The new bound is based on a new second-order moment-closure technique aiming to improve both the computational cost and the accuracy of existing second-order bound. It is equal to the leading eigenvalue of an $(N+M)\times (N+M)$ Metzler matrix depending on the non-backtracking and incidence matrices of the network (Eq.~\eqref{eq:defMathcalA}). Therefore, for sparse networks, the dimension of this matrix grows quasi-linearly. Furthermore, we have shown that the new bound, $\gamma_2$, is tighter than the first-order lower bound, $\gamma_1$, which is equal to the leading eigenvalue of an $N\times N$ matrix depending directly on the adjacency matrix. 

Non-backtracking matrices of networks have been employed for analyzing properties of
stochastic epidemic processes on networks, such as the epidemic threshold of the
SIS model \cite{Shrestha2015PhysRevE,WangTangStanleyBraunstein2017RepProgPhys} and the susceptible-infected-recovered (SIR) model \cite{Karrer2010PhysRevE,Lokhov2015PhysRevE,Morone2015Nature,RadicchiCastellano2016PhysRevE}. The non-backtracking matrix more accurately describes unidirectional state-transition dynamics, such as the SIR dynamics, than the adjacency matrix does because unidirectional dynamics implies that contagions do not backtrack, i.e. if node $v_i$ has infected its neighbor $v_j$, $v_j$ does not re-infect $v_i$.
For the same reason, the non-backtracking matrix also predicts the percolation threshold for networks better than the adjacency matrix \cite{Hamilton2014PhysRevLett,Karrer2014PhysRevLett}. However, the same logic does not apply to the SIS model, in which re-infection through the same edge can happen indefinitely many times. This is a basis of a recent criticism to the application of the non-backtracking matrix to the SIS model \cite{Castellano2018PhysRevE-nonbacktracking}.
For some networks, the epidemic threshold of the SIS model that does not take into account backtracking infection paths 
\cite{Shrestha2015PhysRevE,WangTangStanleyBraunstein2017RepProgPhys} is not accurate \cite{Castellano2018PhysRevE-nonbacktracking}. Although $\gamma_2$ and the corresponding epidemic threshold that we have derived use the non-backtracking matrix, they are mathematical bounds and do not suffer from the inaccuracy caused by the neglect of backtracking infection paths.
%
The present study has shown a new and solid usage of the non-backtracking matrix in understanding the SIS model on networks.

By following the derivation of the epidemic threshold via $\gamma_1$, we derived the epidemic threshold based on $\gamma_2$. The new epidemic threshold is always larger than that based on $\gamma_1$, which is the reciprocal of the largest eigenvalue of the adjacency matrix. Because $\gamma_2$ improves upon $\gamma_1$, we expect that the new epidemic threshold is a better estimate than that based on $\gamma_1$. This point warrants future work. Likewise, the eigenvalue statistics for the adjacency matrix of scale-free networks yield intricate relationships between the epidemic threshold based on $\gamma_1$ and statistics of the node's degree in scale-free networks \cite{Castellano2010PhysRevLett}. How such a result translates to the case of the epidemic threshold based on $\gamma_2$ also warrants future work.

\appendix

\section{\label{sec:Ogura2018 bound}Lower bound on the decay rate derived in Ref.~\cite{Ogura2018SystControlLett}}

In the proof of Theorem~\ref{thm:main}, we have used the following inequality for bounding $q_{ij} = E[x_i(1-x_j)]$ (see Eq.~\eqref{eq:20}): 
\begin{equation}
E[(1-x_i) (1-x_j)x_k] \leq E[(1-x_i) x_k], 
\end{equation}
in which the inequality~$x_j \geq 0$ is used; we have presumed that node $j$ is susceptible. 
In contrast, in our previous study \cite{Ogura2018SystControlLett}, we used
\begin{equation}
E[(1-x_i) (1-x_j)x_k] \leq E[(1-x_j) x_k],
\label{eq:triplet closure Ogura2018}
\end{equation}
which was based on ~$x_i\geq 0$. The use
of Eq.~\eqref{eq:triplet closure Ogura2018} led to the following lower bound on the decay rate~\cite{Ogura2018SystControlLett}:

\begin{theorem} 
Assume that there exist positive numbers $\beta_1$, \dots, $\beta_N$ such that $\beta_{ij} = \beta_j$  for all $i, j \in \{1, \ldots, N\}$. Let $A$ be the adjacency matrix of $\mathcal G$ and its $(i, j)$th entry be $a_{ij}$. Define the $N^2\times N^2$ Metzler matrix 
\begin{equation}
\mathcal B = \begin{bmatrix}
-D & \bigoplus_{i=1}^N (\beta_i A_{i, \backslash\{i\}})
\\
\col_{1\leq i\leq N}(\delta_i V_i) & 
\ \ \bigoplus_{i=1}^N\left(
- \Gamma_i + \col_{j\neq i} \beta_j A_{j, \backslash\{i\}}
\right)
\end{bmatrix}, 
\end{equation}
where $\bigoplus_{i=1}^N M_i$ is the block-diagonal matrix containing matrices $M_1$, $\ldots$, $M_N$ as the diagonal blocks,
$\backslash \{ i\}$ denotes all the columns except the $i$th column,
$V_i\in \mathbb{R}^{(N-1)\times N}$ is the matrix obtained by removing the $i$th row from the $N\times N$ identity matrix, $\Gamma_i  =
\text{diag}(\overline{\gamma}_{i,1}, \ldots, \overline{\gamma}_{i,i-1}, \overline{\gamma}_{i,i+1}, \ldots, \overline{\gamma}_{i,N})$, and
$\overline{\gamma}_{i,j} = \delta_i + \delta_j +
a_{ij}\beta_i$. Then, the decay rate~$\gamma$ satisfies 
\begin{equation}
\gamma \geq \gamma_2^{\prime} \equiv -\lambda_{\max}(\mathcal B). 
\end{equation}
\end{theorem}

\section*{Acknowlegdments}

We thank Claudio Castellano for valuable comments on the manuscript.

\section*{Funding}

National Science Foundation (CAREER-ECCS-1651433 to V.M.P.) and
Japan Society for the Promotion of Science (JP18K13777 to M.O.).

\newpage
\clearpage

\begin{figure}
\includegraphics[width=8cm]{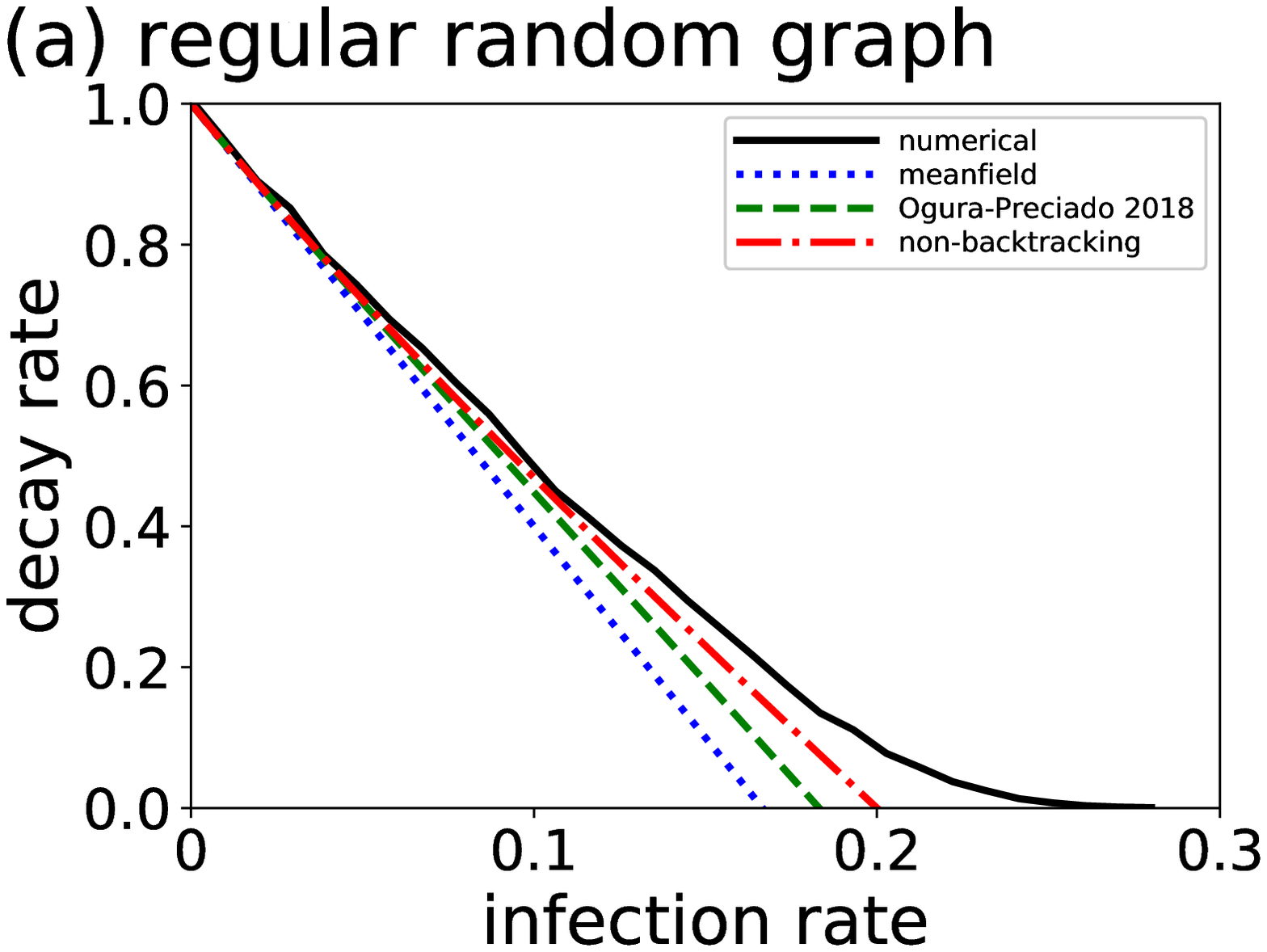}
\includegraphics[width=8cm]{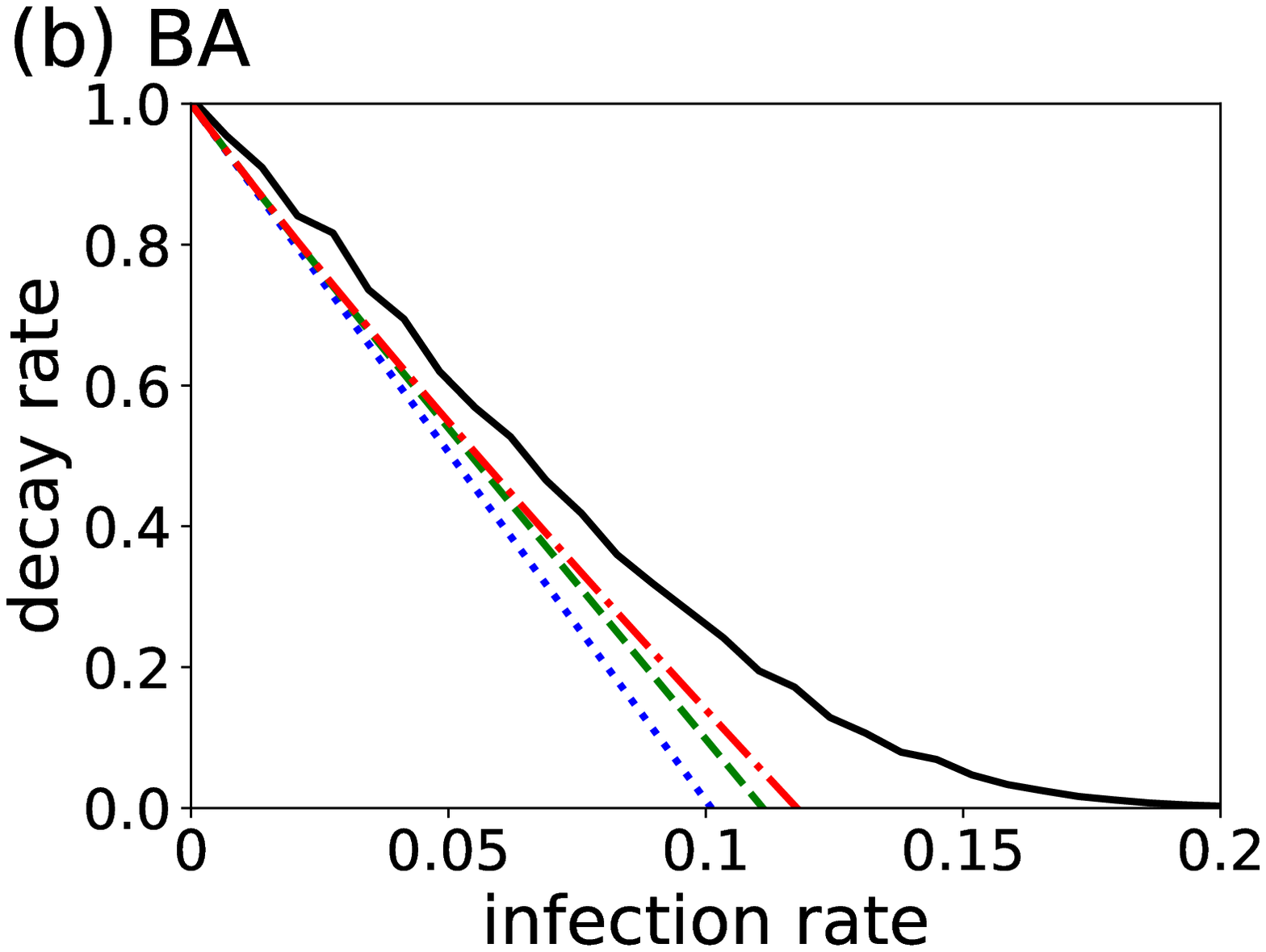}
\includegraphics[width=8cm]{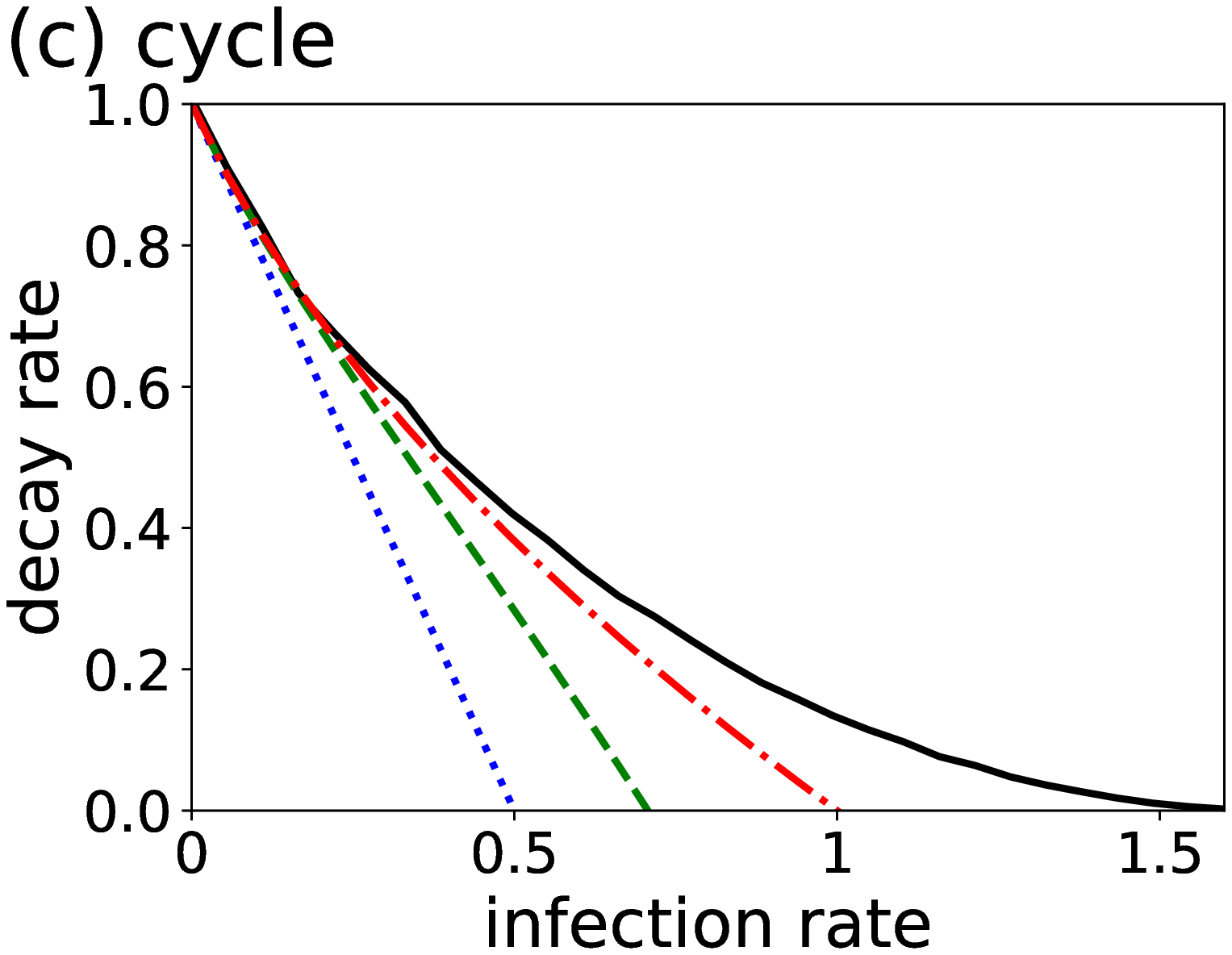}
\includegraphics[width=8cm]{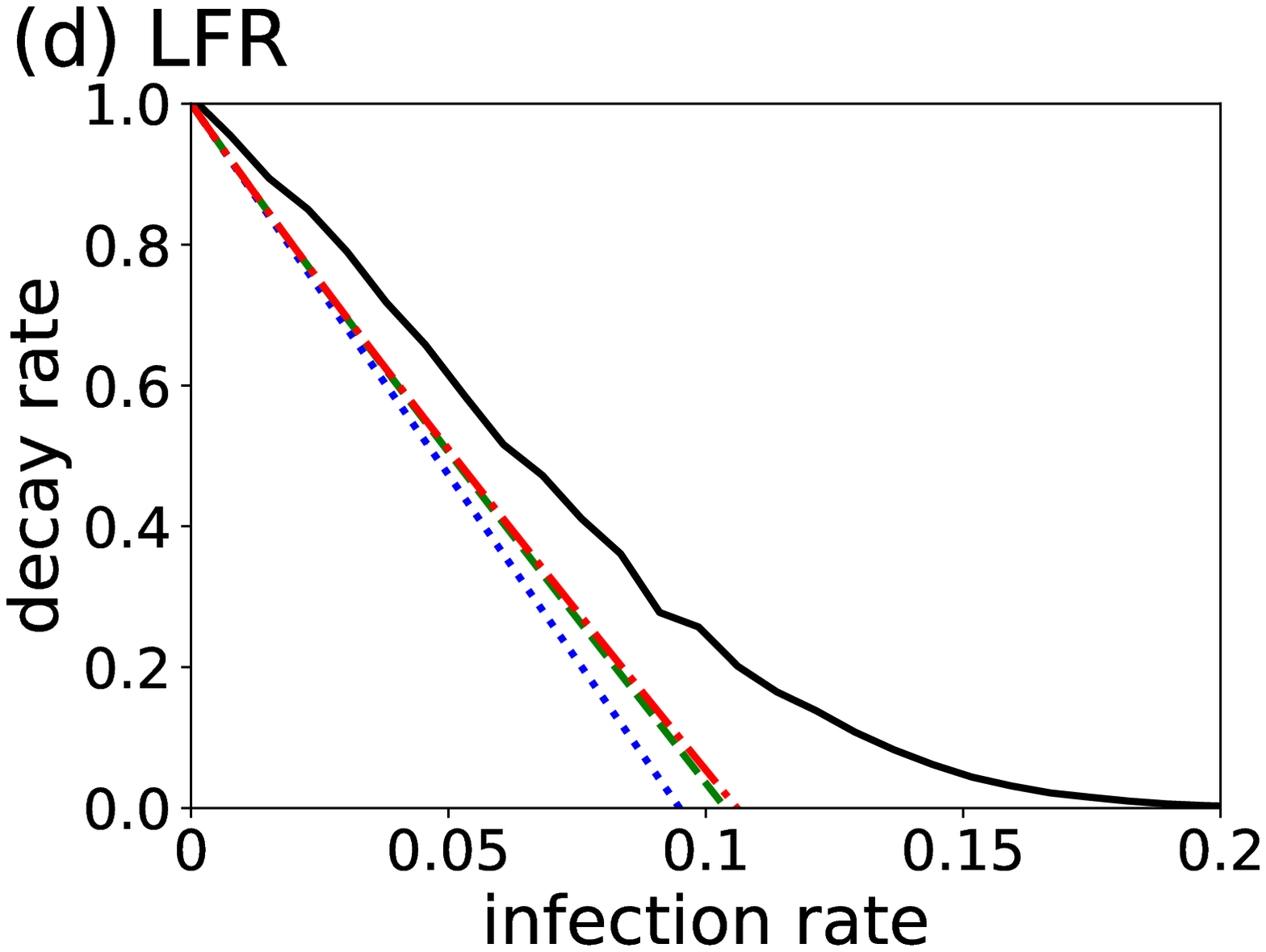}
\caption{Decay rates for different networks and infection rates. (a) Regular random graph. (b) BA model. (c) Cycle. (d) LFR model. (e) Dolphin. (f) Network science. (g) Email. (h) Hamsterster.}
\label{fig:decay rate}
\end{figure}

\clearpage

\noindent
\includegraphics[width=8cm]{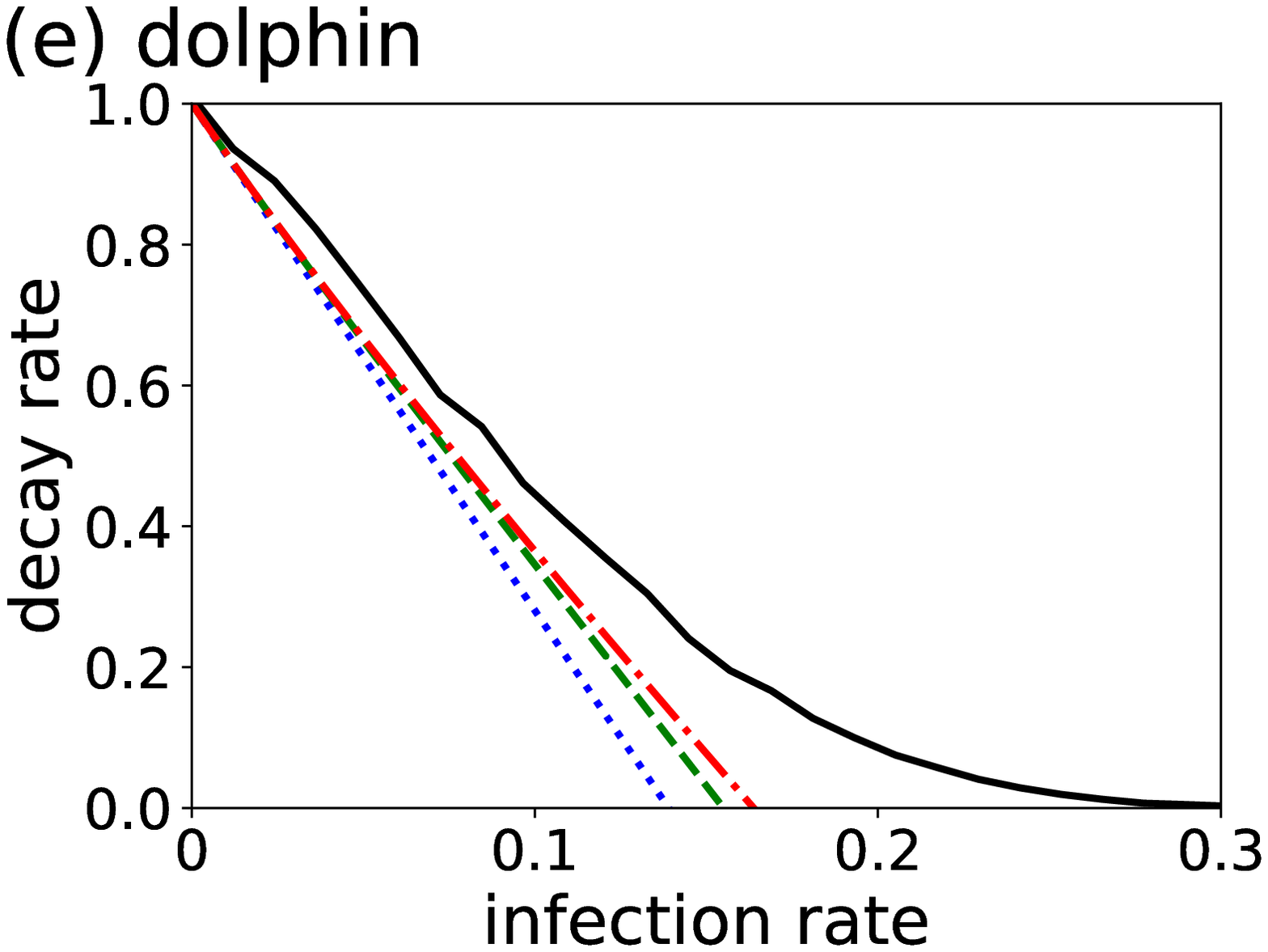}
\includegraphics[width=8cm]{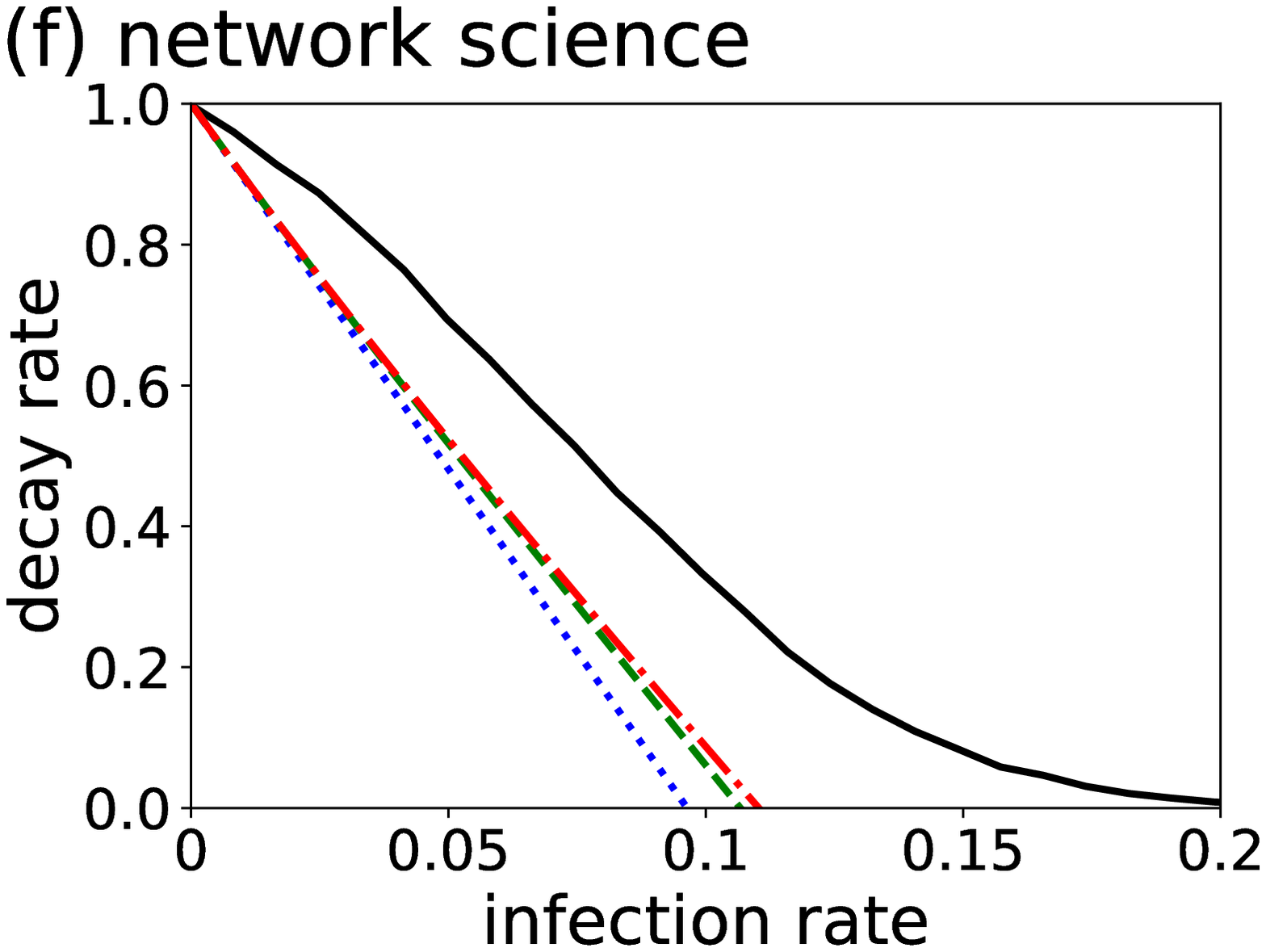}
\includegraphics[width=8cm]{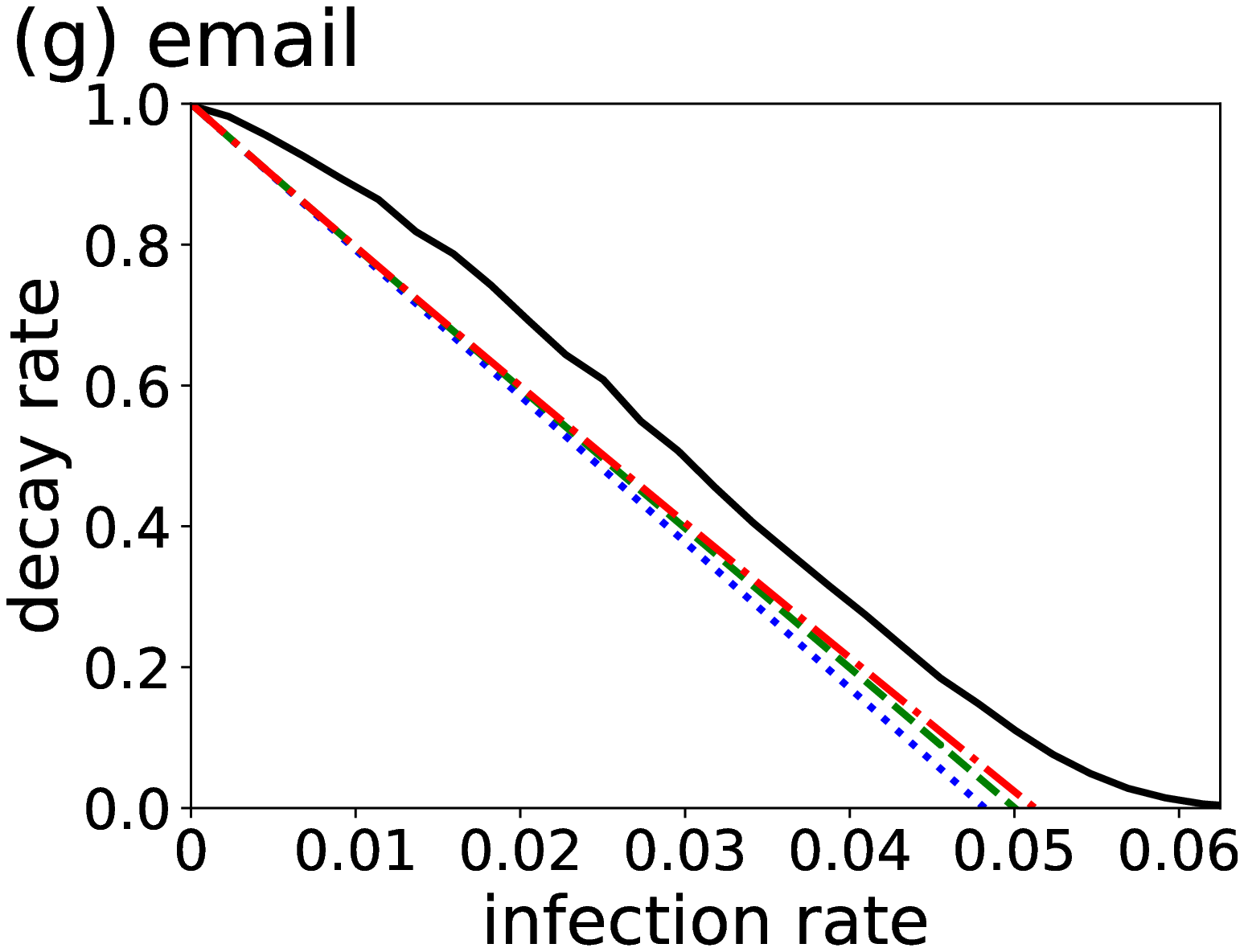}
\includegraphics[width=8cm]{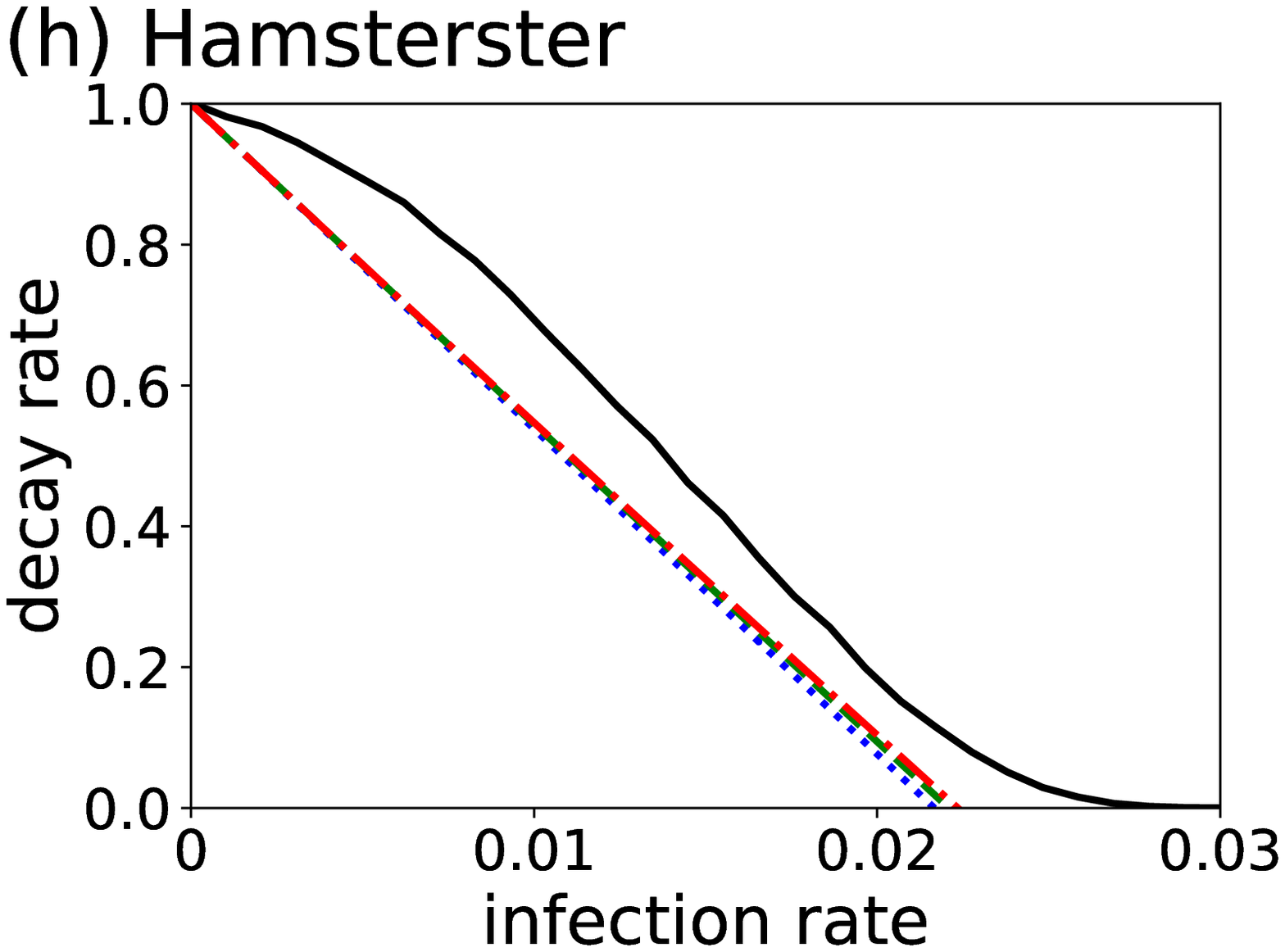}


\end{document}